\documentclass[a4paper,envcountsame,orivec]{llncs}
\usepackage[utf8]{inputenc}
\usepackage[T1]{fontenc}
\usepackage{microtype}
\usepackage[colorlinks]{hyperref}
\usepackage{amsfonts}
\usepackage{amsmath}
\usepackage{amssymb}
\usepackage[small]{complexity}
\usepackage{graphicx}
\usepackage{array}
\usepackage{booktabs}
\usepackage{url}

\newcommand{\N}{\mathbb{N}}
\newcommand{\Reals}{\mathbb{R}}
\newcommand{\DynamicalSystems}{\mathbf{D}}
\newcommand{\Profiles}{\mathbf{P}}
\newcommand{\set}[2][]{#1\{#2#1\}}
\newcommand{\profile}[1]{\mathbf{#1}}
\DeclareMathOperator{\prof}{prof}

\title{Profiles of dynamical systems and their algebra}
\author{Caroline Gaze-Maillot\inst{1} \and Antonio E. Porreca\inst{2}}
\institute{Aix Marseille Université, Université de Toulon, CNRS, LIS, Marseille, France\\\email{caroline.gaze@gmail.com} \and Université Publique\\\email{antonio.porreca@lis-lab.fr}}

\begin{document}

\maketitle

\begin{abstract}
The commutative semiring~$\DynamicalSystems$ of finite, discrete-time dynamical systems was introduced in order to study their (de)composition from an algebraic point of view. However, many decision problems related to solving polynomial equations over~$\DynamicalSystems$ are intractable (or conjectured to be so), and sometimes even undecidable. In order to take a more abstract look at those problems, we introduce the notion of ``topographic'' profile of a dynamical system~$(A,f)$ with state transition function~$f \colon A \to A$ as the sequence~$\prof A = (|A|_i)_{i \in \N}$, where~$|A|_i$ is the number of states having distance~$i$, in terms of number of applications of~$f$, from a limit cycle of~$(A,f)$. We prove that the set of profiles is also a commutative semiring~$(\Profiles,+,\times)$ with respect to operations compatible with those of~$\DynamicalSystems$ (namely, disjoint union and tensor product), and investigate its algebraic properties, such as its irreducible elements and factorisations, as well as the computability and complexity of solving polynomial equations over~$\Profiles$.
\end{abstract}

\section{Introduction}
\label{sec:introduction}

Given a description of a dynamical system, it is often interesting for scientific or engineering purposes to analyse its dynamics, in order to detect its asymptotic behaviour, such as the number and size of limit cycles or fixed points, or other interesting behaviours, such as the reachability of states or its transient paths. However, these problems are often computationally demanding when the system is described in a succinct way, as one normally does, e.g., for Boolean automata networks or cellular automata~\cite{Goles1990a,Sutner1995a}. It is useful, then, to decompose the system into smaller systems before applying such algorithms; if an appropriate decomposition is chosen, the global behaviour of the system may be deduced from the behaviour of its components~\cite{Perrot2020a}.

Let us now consider finite, discrete-time dynamical systems in the most general sense: as finite sets~$A$ of states (including the empty set) together with a state transition function~$f \colon A \to A$. The (countably infinite) set of finite dynamical systems up to isomorphism is a semiring~$(\DynamicalSystems,+,\times)$ with the operations~\cite{Dennunzio2018b}
\begin{align*}
  &(A,f) + (B,g) = (A \uplus B, f + g)&
  &\text{where } (f+g)(x) = \begin{cases}
                              f(x) & \text{if } x \in A \\
                              g(x) & \text{if } x \in B \\
                            \end{cases} \\
  &(A,f) \times (B,g) = (A \times B, f \times g)&
  &\text{where } (f \times g)(a,b) = (f(a), g(b)).
\end{align*}
These operations can also be defined in terms of the graphs of the dynamics as disjoint union~$+$ and graph tensor product~$\times$, which equivalently corresponds to the Kronecker product of the adjacency matrices~\cite{Hammack2011a}.

Given this algebraic structure, one can try to decompose dynamical systems in terms of factoring, or in terms of polynomial equations over~$\DynamicalSystems[\vec{X}]$ in several variables. The decomposition of a dynamical systems in terms of the operations~$+$ and~$\times$ does indeed allow us to detect several interesting dynamical behaviours of the system in terms of its components. For instance, the limit cycles in a sum are just the union of the limit cycles of the addends, while in a product one can predict the number and length of limit cycles as a function of the GCD and LCM of the lengths of the cycles of the factors~\cite{Dorigatti2018a}.

However, solving equations over~$\DynamicalSystems[\vec{X}]$ is not easy either, even if the dynamical systems are given in input explicitly, either as a transition table, or equivalently in terms of the graph of its dynamics~$G(A,f) = (A, \set{(a,f(a)) : a \in A})$. General polynomial equations are even
undecidable~\cite{Dennunzio2018b}, systems of linear equations are~$\NP$-complete, and single equations (even linear) with a constant side are also suspected to be~$\NP$-complete~\cite{Porreca2020a}.

When (de)composing dynamical systems as products, one frequently works starting from the limit cycles and backwards towards the gardens of Eden (states without preimages). It is then useful to know how many states there are at distance~$0, 1, \ldots$ from the limit cycles, as that gives us, for instance, necessary conditions for the compositeness of a system. In this paper we formalise this as the notion of \emph{profile} of a dynamical system, in order to analyse systems from a more abstract point of view. We obtain another semiring~$(\Profiles,+,\times)$ with the ``natural'' operations derived from those of~$(\DynamicalSystems,+,\times)$, analyse some of its algebraic properties (notably, the majority of profiles are irreducible) and prove that working with equivalence classes of systems (with respect to profile equality) ultimately does \emph{not} reduce the complexity of equation problems: general polynomial equations remain undecidable, and even solving a single linear equation is~$\NP$-complete.

\section{Profiles of dynamical systems}
\label{sec:profiles}

Any finite dynamical system~$(A,f)$ consists of one or more disjoint \emph{limit cycles}, which constitute the asymptotic behaviour of the system. Each cycle of length~$1$ is called a \emph{fixed point}, and its only state~$x$ satisfies~$f(x) = x$. The transient (non-asymptotic) behaviour of the system consists of zero or more directed trees of least two nodes having a state of a limit cycle as its root. The existence of limit cycles, which does \emph{not} hold in general for infinite dynamical systems, gives a (pre)ordering to the states, with respect to their distance (in terms of number of applications of~$f$) from the limit cycle in the same connected component of the graph of the dynamics, which we will call its \emph{height}.

\begin{definition}
Let~$(A,f)$ be a dynamical system and let~$x \in A$. We say that the \emph{height of~$x$}, in symbols~$h_A(x)$ or even~$h(x)$ if~$A$ is implied, is the minimum~$h$ such that~$f^h(x)$ is a periodic state, that is, the length of a path from~$x$ to the nearest periodic state in the graph of the dynamics of~$A$.
\end{definition}

It is also possible to generalise the notion of height to a dynamical systems.

\begin{definition}
Let~$(A,f)$ be a dynamical system. We say that the \emph{height of~$A$}, in symbols~$h(A)$, is the maximum height of its states: $h(A) = \max \set{h_A(x) : x \in A}$.
\end{definition}

We can now introduce the notion of \emph{profile} of a dynamical system, whose name is inspired by the topographic profile of a terrain (Fig.~\ref{fig:profile}).

\begin{figure}[t]
\centering
\includegraphics[width=\textwidth]{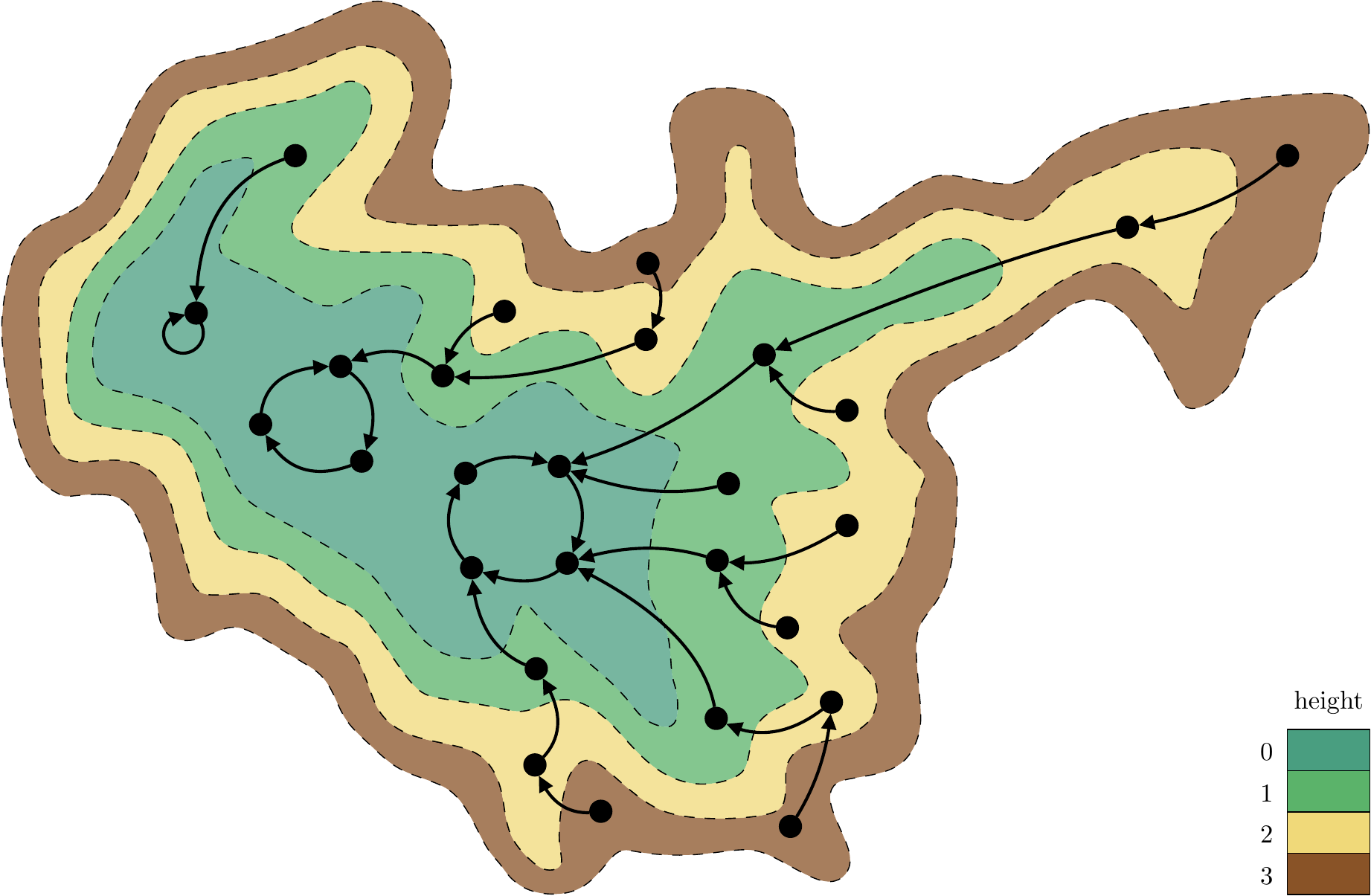}
\caption{Visual representation of the contour lines (isolines) of a dynamical system having profile~$(8,7,8,4)$: there are~$8$ states in limit cycles, $7$~states of height~$1$, $8$~of height~$2$, and~$4$ of height~$3$.}
\label{fig:profile}
\end{figure}

\begin{definition}
Let~$(A,f)$ be a dynamical system of height~$h$, and let~$|A|_i$ be the number of states of~$A$ having height~$i$. This implies~$|A|_i \ge 1$ for~$0 \le i \le h$ and~$|A|_i = 0$ for~$i > h$. Then, the \emph{profile of~$A$} is the eventually null sequence of natural numbers~$\prof (A,f) = (|A|_i)_{i \in \N}$ counting the states of each height in~$A$, in order of height starting from the limit cycles (height~$0$). For brevity, we often write a profile as a finite sequence~$\prof (A,f) = (|A|_0, |A|_1, \ldots, |A|_h)$ by omitting the null terms, except for the profile~$(0)$ of the empty dynamical system.
\end{definition} 

Taking the profile of a dynamical system allows us to work at a higher level of abstraction, since it corresponds to taking a whole equivalence class of dynamical systems. In the rest of this paper, we will denote profiles with bold letters and its elements in italics, such as~$\profile{p} = (p_i)_{i \in \N}$.

\section{The semiring of profiles}
\label{sec:semiring}

It is easy to define a sum operation~$+$ over the (countably infinite) set of profiles~$\Profiles$ that is compatible with the sum over~$\DynamicalSystems$: since~$(A,f) + (B,g)$ is the disjoint union, the profile of this sum is just the elementwise sum of~$\prof(A,f)$ and~$\prof(B,g)$.

\begin{definition}
Given two profiles~$\profile{p} = (p_i)_{i \in \N}$ and~$\profile{q} = (q_i)_{i \in \N}$, define their sum as~$\profile{p} + \profile{q} = (p_i + q_i)_{i \in \N}$.
\end{definition}

\begin{lemma}
\label{thm:profile-of-sum}
For each~$A, B \in \DynamicalSystems$ we have~$\prof(A+B) = \prof A + \prof B$. \qed
\end{lemma}

It is less immediate to define a product over~$\Profiles$, but it is indeed possible with a little more work. First, we show that, in order to compute the height of a state in a product, it suffices to take the maximum height of the two terms.

\begin{lemma}
\label{thm:max-height}
Let~$(A,f), (B,g) \in \DynamicalSystems$ and let~$(C,t) = (A,f) \times (B,g)$. Then, for each~$(a,b) \in C$, we have~$h(a,b) = \max(h(a),h(b))$.
\end{lemma}

\begin{proof}
Notice that the limit cycles of~$C$ consist exactly of the states~$(x,y)$ such that~$y$ is a periodic state of~$A$, and~$y$ a periodic state of~$B$.

Suppose that~$h(a) \ge h(b)$. Then~$f^{h(a)}(a)$ is a periodic state of~$A$, and this is not the case for~$f^i(a)$ whenever~$i < h(a)$, by definition of height. Furthermore, $g^{h(a)}(b)$~is a periodic state of~$B$, since this is the case for all~$g^i(b)$ with~$i \ge h(b)$. Then, $t^{h(a)}(a,b) = (f^{h(a)}(a),g^{h(a)}(b))$ is a periodic state of~$C$, and this is not the case for~$t^i(a,b)$ whenever~$i < h(a)$. This means that~$h(a,b) = h(a)$.

Analogously, if~$h(a) \le h(b)$ we obtain~$h(a,b) = h(b)$, and the statement of the lemma follows. \qed
\end{proof}

The following lemma allows us to count the states of height~$k$ in a product as a function of the number of states of height~$k$ and at most~$k$ of the two factors.

\begin{lemma}
\label{thm:count-points}
Let~$A, B \in \DynamicalSystems$ be dynamical systems and let~$C = A \times B$. Then
\begin{align}
  \label{eq:points-by-height}
  |C|_k = |A|_k \times |B|_{\le k} + |B|_k \times |A|_{\le k} - |A|_k \times |B|_k
\end{align}
for each height~$k$, where~$|D|_{\le k} = |D|_0 + |D|_1 + \cdots + |D|_k$ for each~$D \in \DynamicalSystems$.
\end{lemma}

\begin{proof}
Let~$a \in A$ with~$h(a)=k$ and~$b \in B$ with~$h(b) \le k$. Then, by Lemma~\ref{thm:max-height}, we have~$h(a,b)=k$. This corresponds to the term~$|A|_k \times |B|_{\le k}$ of the sum. Analogously, we have~$h(a,b)=k$ if~$h(a) \le k$ and~$h(b) = k$, which corresponds to the term~$|B|_k \times |A|_{\le k}$. This way, we have counted twice the states~$(a,b)$ with~$h(a)=k$ and~$h(b)=k$, thus it is necessary to subtract the term~$|A|_k \times |B|_k$, and that gives us the correct result. Notice that~\eqref{eq:points-by-height} works even if~$|A|_k = 0$ or~$|B|_k = 0$ (i.e., if~$k > h(A)$ or~$k > h(B)$), giving the expected~$|C|_k = |B|_k \times |A|_{\le k}$ and~$|C|_k = |A|_k \times |B|_{\le k}$, or even~$|C|_k = 0$ if~$k > \max(h(A),h(B))$. \qed
\end{proof}

Notice how Lemma~\ref{thm:count-points} does \emph{not} depend on the exact shapes of~$A$ and~$B$, but only on their profile. This allows us to define the product of profiles as follows~(Fig.~\ref{fig:products}):

\begin{figure}[t]
\centering
\renewcommand{\arraystretch}{2}
\newcolumntype{C}{>{\centering\arraybackslash}m{3.5em}}
\begin{tabular}{C|*{8}{C}}
  $\times$ & (0) & (1) & (2) & (1,1) & (3) & (1,2) & (2,1) & (1,1,1) \\
  \hline
  (0) & (0) & (0) & (0) & (0) & (0) & (0) & (0) & (0) \\
  (1) & (0) & (1) & (2) & (1,1) & (3) & (1,2) & (2,1) & (1,1,1) \\
  (2) & (0) & (2) & (4) & (2,2) & (6) & (2,4) & (4,2) & (2,2,2) \\
  (1,1) & (0) & (1,1) & (2,2) & (1,3) & (3,3) & (1,5) & (2,4) & (1,3,2) \\
  (3) & (0) & (3) & (6) & (3,3) & (9) & (3,6) & (6,3) & (3,3,3) \\
  (1,2) & (0) & (1,2) & (2,4) & (1,5) & (3,6) & (1,8) & (2,7) & (1,5,3) \\
  (2,1) & (0) & (2,1) & (4,2) & (2,4) & (6,3) & (2,7) & (4,5) & (2,4,3) \\
  (1,1,1) & (0) & (1,1,1) & (2,2,2) & (1,3,2) & (3,3,3) & (1,5,3) & (2,4,3) & (1,3,5)
\end{tabular}
\caption{Multiplication table for profiles of size~$0$, $1$, $2$, and~$3$.}
\label{fig:products}
\end{figure}

\begin{definition}
\label{def:product}
For two profiles~$\profile{p} = (p_i)_{i \in \N}$ and~$\profile{q} = (q_i)_{i \in \N}$, let their product be
\begin{align*}
  \textstyle
  \profile{p} \times \profile{q} = \Big(
    p_i \times \sum_{j=0}^i q_i + q_i \times \sum_{j=0}^i p_i - p_i \times q_i
  \Big)_{i \in \N}.
\end{align*}
\end{definition}

From Definition~\ref{def:product} and Lemma~\ref{thm:count-points} we obtain the expected result:

\begin{lemma}
\label{thm:profile-of-product}
For each~$A, B \in \DynamicalSystems$ we have~$\prof(A \times B) = \prof A \times \prof B$. \qed
\end{lemma}

This finally gives us the algebraic structure of~$\Profiles$.

\begin{theorem}
$(\Profiles, +, \times)$ is a commutative semiring.
\end{theorem}

\begin{proof}
The operation~$+$ inherits the associative and commutative properties from~$\N$ and has, as the neutral element, the null profile~$0_\Profiles = (0)$, i.e., the profile of the empty dynamical system~$0_\DynamicalSystems = \varnothing$. Thus~$(\Profiles,+)$ is a commutative monoid.

Let~$\profile{p}, \profile{q}, \profile{r} \in \Profiles$, and let~$A,B,C \in \DynamicalSystems$ such that~$\prof A = \profile{p}, \prof B = \profile{q}$, and~$\prof C = \profile{r}$. Then we have~$\profile{p} \times \profile{q} = \prof A \times \prof B = \prof(A \times B)$ by Lemma~\ref{thm:profile-of-product}, then~$\prof(A \times B) = \prof(B \times A)$ by commutativity of~$\times$ in~$\DynamicalSystems$, and~$\prof(B \times A) = \prof B \times \prof A = \profile{q} \times \profile{p}$, and thus~$\times$ is commutative. Similarly, we have the associative property~$\profile{p} \times (\profile{q} \times \profile{r}) = (\profile{p} \times \profile{q}) \times \profile{r}$, the neutral element~$1_\Profiles = \prof 1_\DynamicalSystems = (1)$, i.e., the profile of a the dynamical system consisting of a single fixed point, and the distributive property~$\profile{p} \times (\profile{q} + \profile{r}) = \profile{p} \times \profile{q} + \profile{p} \times \profile{r}$. \qed
\end{proof}

From Lemmata~\ref{thm:profile-of-sum} and~\ref{thm:profile-of-product} we also obtain that ``taking the profile'' does indeed respect the semiring operations and, being surjective, it gives us, in a sense, a good abstraction of dynamical systems.

\begin{corollary}
The function~$\prof \colon \DynamicalSystems \to \Profiles$ is a surjective semiring homomorphism. \qed
\end{corollary}

A very important and useful result is that~$\Profiles$ contains an isomorphic copy of the naturals, the initial semiring of its category:

\begin{lemma}
\label{thm:profiles-contains-nats}
$(\Profiles,+,\times)$ has a subsemiring isomorphic to~$(\N,+,\times)$.
\end{lemma}

\begin{proof}
Let~$\phi \colon \N \to \DynamicalSystems$ be defined by~$\phi(n) = (n)$, that is, the profile having~$n$ as its first component and zero everywhere else. Then~$\phi(0) = (0) = 0_\Profiles$, $\phi(1) = (1) = 1_\Profiles$, $\phi(m+n) = (m+n) = (m) + (n) = \phi(m) + \phi(n)$, and~$\phi(m \times n) = (m \times n) = (m) \times (n) = \phi(m) \times \phi(n)$. Thus~$\phi$ is a semiring homomorphism. Furthermore, $\phi(m) = \phi(n)$~implies~$m = n$, i.e., $\phi$~is injective. As a consequence, its image~$\phi(\N)$ is a subsemiring of~$\Profiles$ isomorphic to~$\N$. \qed
\end{proof}

The size of a profile is the number of states of any dynamical system with that profile, and it also enjoys some nice properties.

\begin{definition}
The \emph{size} of a profile~$\profile{p} = (p_i)_{i \in \N}$ is given by~$|\profile{p}| = \sum_{i \in \N} p_i$.
\end{definition}

\begin{lemma}
\label{thm:size-is-hom}
The function~$|\cdot| \colon \Profiles \to \N$ is a semiring homomorphism.
\end{lemma}

\begin{proof}
The size~$|\profile{p}|$ of a profile~$\profile{p}$ is just the number of states~$|A|$ of any dynamical system~$A$ such that~$\prof A = \profile{p}$. Since the sum and product of dynamical systems have the disjoint union and the Cartesian product as their set of states, respectively, we have~$|A+B| = |A| + |B|$ and $|A \times B| = |A| \times |B|$ for~$A,B \in \DynamicalSystems$. Furthermore, we have~$|0_\Profiles| = 0$ and~$|1_\Profiles| = 1$. The result then follows from Lemmata~\ref{thm:profile-of-sum} and~\ref{thm:profile-of-product}. \qed
\end{proof}

Since the profiles~$\Profiles$ contain the naturals (Lemma~\ref{thm:profiles-contains-nats}), they are not only a semiring, but also an~$\N$-semimodule, a ``vector space'' with the naturals as its ``scalars''~\cite{Hebisch1998a}, with the semimodule axioms satisfied as a direct consequence of the semiring axioms. This will be useful later when analysing the complexity of solving linear equations over~$\Profiles$ (Section~\ref{sec:equations}).

\begin{theorem}
\label{thm:profiles-semimodule}
$(\Profiles, +)$ is an $\N$-semimodule with its ordinary multiplication restricted to~$\N \times \Profiles \to \Profiles$, that is, $n \times (p_i)_{i \in \N} = (np_i)_{i \in \N}$. \qed
\end{theorem}

Unfortunately, profiles are not particularly nice as a semimodule, since there is only one minimal generating set, and it is not linearly independent.

\begin{theorem}
\label{thm:profiles-generators}
$\Profiles$ as an~$\N$-semimodule has a unique, countably infinite minimal generating set~$G = \set{\profile{p} \in \Profiles : p_0 = 1}$, the set of profiles starting with~$1$, which is linearly dependent.
\end{theorem}

\begin{proof}
The set~$G$ is a generating set, since any profile~$\profile{q} = (q_0, q_1, q_2, \ldots)$ can be written as~$(q_0-1)\times(1) + 1 \times (1, q_1, q_2, \ldots)$, that is, any element of~$\Profiles$ is a linear combination of at most two elements of~$G$. This set is countably infinite.

To prove that any generating set of~$\Profiles$ must contain~$G$, consider any element~$\profile{p} \in G$. If it is a linear combination~$\profile{p} = \sum_{i=1}^m a_i \profile{q}_i$ of profiles~$\profile{q}_i \in \Profiles$ with coefficients~$a_i \in \N$, then one of the~$\profile{q}_i$ must start with~$1$ (i.e., $q_{i,0} = 1$) and its coefficient~$a_i$ must be~$1$ as well, otherwise we would have either~$p_0 = 0$ or~$p_0 > 1$; furthermore, we must have~$(a_j \profile{q}_j)_0 = 0$ for all~$j \ne i$, which implies~$a_j \profile{q}_j = 0$ (since all elements of a profile are null after the first~$0$) and thus~$a_j = 0$. But then we have~$\profile{p} = 1 \profile{q}_i$, that is, $\profile{q}_i = \profile{p}$: any linear combination giving~$\profile{p}$ as its result must contain~$\profile{p}$ itself, and thus it must belong to any generating set.

However, $G$ is not linearly independent, that is, there exist~$m$ profiles~$\profile{p}_i \in \Profiles$ and corresponding natural numbers~$a_i,b_i \in \N$ such that~$\sum_{i=1}^m a_i \profile{p}_i = \sum_{i=1}^m b_i \profile{p}_i$ but~$a_i \ne b_i$ for some~$i$. For instance, we have~$(1, 1) + (1, 2) = (1) + (1, 3)$.
\qed
\end{proof}

\section{Reducibility and factorisation of profiles}
\label{sec:reducibility}

When studying the reducibility of profiles with respect to semiring product, a simple sufficient condition for irreducibility is given by the primality of its size.

\begin{lemma}
\label{thm:prime-size-irreducible}
Let~$\profile{p} \in \Profiles$ be a profile such that~$|\profile{p}|$ is prime. Then~$\profile{p}$ is irreducible.
\end{lemma}

\begin{proof}
Suppose~$\profile{p} = \profile{q} \times \profile{r}$. Since~$|\cdot|$ is a semiring homomorphism (Lemma~\ref{thm:size-is-hom}), we have~$|\profile{p}| = |\profile{q}| \times |\profile{r}|$. But~$|\profile{p}|$ is prime, thus either~$|\profile{q}| = 1$, or~$|\profile{r}| = 1$. Since the only profile of size~$1$ is~$1_\Profiles = (1)$, this factorisation is trivial. \qed
\end{proof}

It is easy to check by inspection of the product table (Fig.~\ref{fig:products}) that some profiles admit multiple factorisations into irreducibles, a property that they share with the semiring of dynamical systems~\cite{Dennunzio2018b}.

\begin{theorem}
$\Profiles$ is \emph{not} a unique factorisation semiring.
\end{theorem}

\begin{proof}
The smallest counterexample is the profile~$(2,4)$, which has two distinct factorisations:~$(2,4) = (2) \times (1,2) = (1,1) \times (2,1)$. All the factors~$(2)$, $(1,2)$, $(1,1)$, and~$(2,1)$ are irreducible because of their prime size (Lemma~\ref{thm:prime-size-irreducible}). \qed
\end{proof}

Another property in common with dynamical systems~\cite{Dorigatti2018a} is that most profiles, a fraction asymptotically equal to~$1$, are irreducible.

\begin{theorem}
\label{thm:majority-irreducible}
The majority of profiles is irreducible; specifically,
\begin{align*}
  \lim_{n \to \infty} \frac{\text{number of reducible profiles of size at most~$n$}}{\text{number of profiles of size at most~$n$}} = 0.
\end{align*}
\end{theorem}

\begin{proof}
There are as many profiles of size~$i$ as there are ordered tuples of strictly positive naturals having sum~$i$, which correspond to the ways of writing~$i$ as an ordered sum of strictly positive integers. The latter are the \emph{compositions} of~$i$, and there are~$2^{i-1}$ of them for~$i \ge 1$~\cite[Chapter~4]{Andrews1984a}, and~$1$ for~$i=0$. Hence, the number of profiles of size at most~$n$ is given by~$1 + \sum_{i=1}^n 2^{i-1} = 1 + 2^n - 1 = 2^n$.

Suppose that~$\profile{p} \in \Profiles$ has size~$i = |\profile{p}|$ and that~$i = \ell m$. Then, there are \emph{at most}~$2^{\ell-1} \times 2^{m-1} = 2^{\ell+m-2}$ ways of choosing profiles~$\profile{q}$ and~$\profile{r}$, of sizes~$\ell$ and~$m$ respectively, such that~$\profile{p} = \profile{q} \times \profile{r}$.

Let~$k$ be the number of distinct factorisations of~$i$ into products of two integers~$\ell_j \ge m_j > 1$. The number of ways of decomposing~$\profile{p}$ into a product of two non-trivial divisors is at most~$\sum_{j=1}^k 2^{\ell_j+m_j-2}$. Observe\footnote{This can be proved by induction on any of the two variables.} that~$\ell+m \le \ell m / 2 + 2$ for all~$\ell,m>1$; this implies~$\sum_{j=1}^k 2^{\ell_j+m_j-2} \le \sum_{j=1}^k 2^{\ell_j m_j / 2} = \sum_{j=1}^k 2^{i/2} = k \sqrt{2^i}$. We have~$k < i$, since the number of non-trivial divisors of~$i$ is strictly less than~$i$ (at least~$1$ and~$i$ have to be thrown out), hence the number of ways of obtaining a profile of size~$i$ as a product of two non-trivial profiles, which is the same as the number of reducible profiles of size~$i$, is bounded by~$i \sqrt{2^i}$ for~$k \ge 1$, and it is~$1$ for~$i=0$. The number of reducible profiles of size at most~$n$ is then bounded by~$1 + \sum_{i=1}^n i \sqrt{2^i} \le 1 + n \sum_{i=0}^n \sqrt{2^i} = 1 + n \frac{\sqrt{2^{n+1}} - 1 }{\sqrt{2} - 1}$.

By dividing that by the number of profiles of size~$n$ we obtain
\begin{align*}
  \lim_{n \to \infty}
  \frac{1 + n \frac{\sqrt{2^{n+1}} - 1 }{\sqrt{2} - 1}}{2^n} = 0 
\end{align*}
as required. \qed
\end{proof}

Thus, the semiring of profiles is quite complex from the point of view of reducibility: most profiles are not reducible at all, but those that are sometimes admit multiple factorisations. Furthermore, since height-$1$ profiles behave as the natural numbers, we also obtain a complexity lower bound to profile factorisation.

\begin{theorem}
The problem of profile factorisation, that is, given a profile~$\profile{p} \in \Profiles$, finding a divisor~$\profile{d}$ of~$\profile{p}$ with~$\profile{d} \ne 1_\Profiles$ and~$\profile{d} \ne \profile{p}$ (or answering that~$\profile{p}$ is irreducible, if this is the case) is at least as hard as integer factorisation.
\end{theorem}

\begin{proof}
Given a natural number~$n$, let us consider the profile~$\profile{n} = (n)$, that is, $n$ followed by zeros. This profile can only be divided by profiles~$\profile{d} = (d)$ and~$\profile{q} = (q)$ of height~$0$ by Lemma~\ref{thm:max-height}. Then, $\profile{n} = \profile{d} \times \profile{q} = (d \times q + d \times q - d \times q) = (d \times q)$ for some profile~$\profile{d}$ with~$\profile{d} \ne 1_\Profiles$ and~$\profile{d} \ne \profile{n}$ if and only if~$n = d \times q$ for some~$d$ with~$d \ne 1$ and~$d \ne n$, which is the integer factorisation problem. \qed
\end{proof}

\section{Solving polynomial equations over profiles}
\label{sec:equations}

One of the reasons for introducing profiles is to abstract away from the exact shape of dynamical systems, with the hope of making polynomial equations easier to solve. As we will show in this section, this is not at all the case.
First of all, let us prove that polynomial equations with \emph{natural} coefficients do sometimes have \emph{non-natural} solutions in~$\Profiles$ (e.g.,~$3X = Z$ has solution~$X = (1,2), Z = (3,6)$), but only if there also exist natural ones (e.g., $X=3, Z=9$), as in the semiring~$\DynamicalSystems$~\cite{Dennunzio2018b}.

\begin{lemma}
Let~$p,q \in \N[\vec{X}]$ be polynomials with natural coefficients over the variables~$\vec{X} = (X_1, \ldots, X_m)$. Then the equation~$p(\vec{X}) = q(\vec{X})$ has a solution in~$\Profiles$ if and only if it has a solution in~$\N$.
\end{lemma}

\begin{proof}
If the equation has a solution in~$\N$, then this is already a solution in~$\Profiles$. Conversely, let~$\vec{\profile{r}} = (\profile{r}_1, \ldots, \profile{r}_m)$ be a solution in~$\Profiles$. We claim that~$|\vec{\profile{r}}| = (|\profile{r}_1|, \ldots, |\profile{r}_m|)$ is also a solution, in~$\N$; that is, by replacing each profile by (the dynamical system corresponding to) its size, the equation remains valid.

If the equation~$p(\vec{X}) = q(\vec{X})$ is of degree at most~$n$, then it can be written as
\begin{math}
  \sum_{\vec{i} \in [0,n]^m} \big(
    a_{\vec{i}} \prod_{j=1}^m X_j^{i_j}
  \big) =
  \sum_{\vec{i} \in [0,n]^m} \big(
    b_{\vec{i}} \prod_{j=1}^m X_j^{i_j}
  \big),
\end{math}
that is, we compute all products of the~$m$ variables, each variable with an exponent ranging from~$0$ to~$n$ (these exponents are collected in a vector~$\vec{i} \in [0,n]^m$), and multiply it by a corresponding coefficient~$a_{\vec{i}} \in \N$ or~$b_{\vec{i}} \in \N$, and then all these monomials are added together. Any of the coefficients~$a_{\vec{i}}$ and~$b_{\vec{i}}$ can be~$0$ (if there is more than one variable, some of them will surely be, in order to keep the degree at most~$n$).

If~$\vec{\profile{r}}$ is a solution the equation, i.e., if~$p(\vec{\profile{r}}) = q(\vec{\profile{r}})$, then by expanding we obtain
\begin{math}
  \sum_{\vec{i} \in [0,n]^m} \big(
    a_{\vec{i}} \prod_{j=1}^m \profile{r}_j^{i_j}
  \big) =
  \sum_{\vec{i} \in [0,n]^m} \big(
    b_{\vec{i}} \prod_{j=1}^m \profile{r}_j^{i_j}
  \big).
\end{math}
By applying the size function~$|\cdot|$ to both sides of the equation, and exploiting the fact that it is a semiring homomorphism (Lemma~\ref{thm:size-is-hom}) and that~$|a_{\vec{i}}| = a_{\vec{i}}$ and~$|b_{\vec{i}}| = b_{\vec{i}}$ since they already are natural numbers, we obtain the equation over the naturals
\begin{math}
  \sum_{\vec{i} \in [0,n]^m} \big(
    a_{\vec{i}} \prod_{j=1}^m |\profile{r}_j|^{i_j}
  \big) =
  \sum_{\vec{i} \in [0,n]^m} \big(
    b_{\vec{i}} \prod_{j=1}^m |\profile{r}_j|^{i_j}
  \big),
\end{math}
which is nothing else than~$p(|\vec{\profile{r}}|) = q(|\vec{\profile{r}}|)$. Thus~$|\vec{\profile{r}}|$ is indeed a natural solution to the original equation.
\qed
\end{proof}

As a consequence, ``Hilbert's 10th problem over~$\Profiles$'' has a negative answer: there is no algorithm for deciding if a polynomial equation in~$\Profiles[\vec{X}]$ is solvable, otherwise you could use the same algorithm for natural equations.

\begin{theorem}
Deciding whether an equation~$p(\vec{X}) = q(\vec{X})$ with~$p, q \in \Profiles[\vec{X}]$ has a solution in~$\Profiles$ is undecidable. \qed
\end{theorem}

We can get a subclass of algorithmically solvable equations by having one constant side, that, by considering equations of the form~$p(\vec{X}) = \profile{q}$ with~$\profile{q} \in \Profiles$. The constant side makes the search space of the solutions finite, which means that at least a brute-force search algorithm is available.

\begin{lemma}
\label{thm:operation-bound}
Let~$\profile{p}, \profile{q}, \profile{r} \in \Profiles$ be profiles. Then~$\profile{p} + \profile{q} = \profile{r}$ implies~$p_i \le r_i$, and~$\profile{p} \times \profile{q} = \profile{r}$ implies~$p_i \le r_i$ whenever~$\profile{q} \ne 0_\Profiles$, for all~$i \in \N$.
\end{lemma}

\begin{proof}
If~$\profile{p} + \profile{q} = \profile{r}$, then~$p_i \le p_i + q_i = r_i$ for all~$i \in \N$, as required. Now suppose~$\profile{p} \times \profile{q} = \profile{r}$ and~$\profile{q} \ne 0_\Profiles$. By Definition~\ref{def:product}, this means
\begin{align*}
  r_i = (\profile{p} \times \profile{q})_i =
    p_i \times \sum_{j=0}^i q_i + q_i \times \sum_{j=0}^i p_i - p_i \times q_i.
\end{align*}
Since~$\profile{q} \ne 0_\Profiles$, we have~$q_j \ge 1$ for at least one~$j \le i$. This means that~$p_i \times \sum_{j=0}^i q_i \ge p_i$. If~$q_i = 0$ then~$r_i \ge p_i$ as required. So suppose~$q_i \ge 1$; this implies~$q_i \times \sum_{j=0}^i p_i \ge q_i \times p_i$ and~$r_i \ge p_i + q_i \times p_i - p_i \times q_i = p_i$, which completes the proof. \qed
\end{proof}

By applying Lemma~\ref{thm:operation-bound} repeatedly, we obtain the following result.

\begin{lemma}
\label{thm:polynomial-bound}
Let~$p(\vec{X}) \in \Profiles[\vec{X}]$ over the variables~$\vec{X} = (X_1, \ldots, X_m)$, and let~$\profile{q} \in \Profiles$ be a constant. Then, if~$p(\vec{\profile{r}}) = \profile{q}$ for some~$\vec{\profile{r}} = (\profile{r}_1, \ldots, \profile{r}_m) \in \Profiles^m$, there exists a (possibly different) solution~$\vec{\profile{s}} = (\profile{s}_1, \ldots, \profile{s}_m) \in \Profiles^m$ such that~$p(\vec{\profile{s}}) = \profile{q}$ and~$s_{i,j} \le q_j$ for all~$1 \le i \le m$ and~$j \in \N$.
\end{lemma}

\begin{proof}
If all coefficients of~$p$ are nonzero, and all profiles~$\profile{r}_i$ are also nonzero, then let~$\vec{\profile{s}} = \vec{\profile{r}}$, and the result follows from Lemma~\ref{thm:operation-bound} by induction on the structure of the expression~$p(\vec{\profile{r}})$.

Otherwise, the expression~$p(\vec{\profile{r}})$ is a sum with at least one null term, say~$\profile{a} \times \profile{p}_{i_1}^{e_1} \times \cdots \times \profile{p}_{i_k}^{e_k}$. If any of the terms~$\profile{p}_i$ occurring in this product have~$p_{i,j} > q_j$ for some~$j \in \N$, then it means that~$\profile{p}_i$ never occurs in a non-null term of the expression~$p(\vec{\profile{r}})$, since all sums are computed elementwise and this would invalidate the equality. Hence~$\profile{p}_i$ only occurs multiplied by~$0_\Profiles$, and it can thus be replaced by \emph{any} profile~$\profile{p}_i'$ satisfying~$p'_{i,j} \le q_j$ (for instance, $\profile{p}_i' = 0_\Profiles$ always works) without changing the validity of the equation. By repeating this operation with all profiles~$\profile{p}_i$ of this kind, we obtain another solution~$\vec{\profile{s}}$ which satisfies the required inequalities. \qed
\end{proof}

In the rest of the paper we encode profiles, as is natural, as finite sequences of natural numbers~$\profile{p} = (p_0, \ldots, p_h)$ in binary notation. We can prove that polynomial equation with a constant side can be solved in nondeterministic polynomial time.

\begin{lemma}
\label{thm:constant-rhs-np}
Deciding whether an equation~$p(\vec{X}) = \profile{q}$, with~$p \in \Profiles[\vec{X}]$ and constant right-hand side~$\profile{q} \in \Profiles$, has a solution in~$\Profiles$ is an~$\NP$ problem. The same holds for a \emph{system} of equations.
\end{lemma}

\begin{proof}
By Lemma~\ref{thm:polynomial-bound}, the equation has a solution if and only if it has a solution~$\vec{\profile{r}} = (\profile{r}_1, \ldots, \profile{r}_m)$ where each element of~$\profile{r}_1, \ldots, \profile{r}_m$ is bounded by an element of~$\profile{q}$. Thus, guessing a solution to the equation amounts to guessing, for each~$\profile{r}_i = (r_{i,0}, \ldots, r_{i,h})$, a natural number~$r_{i,j} \in [0,q_j]$ for each height~$0, \ldots, h(\profile{q})$. This can be performed in nondeterministic polynomial time. Then, the candidate solution can be verified in deterministic polynomial time by evaluating the~$p(\vec{X})$ in~$\vec{\profile{r}}$ and checking equality with~$\profile{q}$. This proves that the problem belongs to~$\NP$.

In the case of multiple equations, after guessing the solution we need to verify that \emph{all} equations are satisfied, which still takes polynomial time. \qed
\end{proof}

Unfortunately, the~$\NP$ upper bound is strict. We prove that first for systems of linear equations.

\begin{theorem}
\label{thm:linear-equations}
Deciding whether a \emph{system of linear equations}
\begin{align*}
  \underbrace{
    p_1(\vec{X}) = \profile{q}_1
    \qquad \cdots \qquad
    p_n(\vec{X}) = \profile{q}_n
  }
\end{align*}
with~$p_i \in \Profiles[\vec{X}]$ and constant right-hand sides~$\profile{q}_i \in \Profiles$ has a solution in~$\Profiles$ is~$\NP$-complete.
\end{theorem}

\begin{proof}
We prove that the problem is $\NP$-hard by reduction from the~$\NP$-complete problem One-in-three 3SAT, the problem of deciding whether a Boolean formula~$\varphi$ in ternary conjunctive normal form has a satisfying assignment which makes only one literal per clause true~\cite{Schaefer1978a}.

For each logical variable~$x$ of~$\varphi$ we have a pair of variables~$X$ and~$X'$ and an equation~$X + X' = 1$. This equation forces exactly one between~$X$ and~$X'$ to~$1$, and the other to~$0$. We use~$X$ to represent~$x$ and~$X'$ to represent~$\neg x$.

For each clause~$(\ell_1 \lor \ell_2 \lor \ell_3)$ of three literals we have an equation~$L_1 + L_2 + L_3 = 1$, where~$L_i = X_i$ if~$\ell_i = x_i$, and~$L_i = X_i'$ if~$\ell_i = \neg x_i$. This forces exactly one variable corresponding to a literal of the clause to~$1$, and the other two to~$0$.

Then, the system of equations obtained from~$\varphi$ is linear and has constant right-hand sides; furthermore, it has a solution if and only if the formula has a satisfying assignment. The satisfying assignments and the solutions to the system of equations are actually the same, if we interpret~$0$ as false and~$1$ as true.
\qed
\end{proof}

By exploiting the~$\N$-semimodule structure of~$\Profiles$, we can combine several linear equations together, proving that even a single one is already~$\NP$-hard.

\begin{theorem}
\label{thm:one-equation-np-complete}
Deciding whether a \emph{single linear equation}~$p(\vec{X}) = \profile{q}$, with~$p \in \Profiles[\vec{X}]$ and constant right-hand side~$\profile{q} \in \Profiles$, has a solution in~$\Profiles
$ is~$\NP$-complete.
\end{theorem}

\begin{proof}
We prove this problem~$\NP$-complete by adapting the proof of Theorem~\ref{thm:linear-equations}, reducing the system of linear equations~$p_1(\vec{X}) = 1, \ldots, p_n(\vec{X}) = 1$ to a single linear equation. Remark that all coefficients of the polynomials~$p_i$, as well as the right-hand sides of the equations, are actually natural numbers in that proof.

As mentioned above (Theorem~\ref{thm:profiles-semimodule}), $\Profiles$ is an~$\N$-semimodule. Consider the elements~$\profile{e}_i = \underbrace{(1, \ldots, 1)}_{i \text{ times}} \in \Profiles$ for~$1 \le i \le n$; these elements are linearly independent over~$\N$, since the $n \times n$ matrix over~$\Reals$ having the length-$n$ prefixes of~$\profile{e}_1, \ldots, \profile{e}_n$ as columns has determinant~$1$, and if there exists no linear combination with nonzero real coefficients giving the null vector, certainly there does not exist one with natural coefficients. Now consider the following equation:
\begin{align*}
  \sum_{i=1}^n \profile{e}_i p_i(\vec{X}) = \sum_{i=1}^n \profile{e}_i.
\end{align*}
It is a linear equation with constant right-hand side in~$\Profiles$ and left-hand side in~$\Profiles[\vec{X}]$. By linear independence over~$\N$ of the elements~$\profile{e}_i$, this equality holds if and only if~$p_i(\vec{X}) = 1$ for all~$1 \le i \le n$, that is, the equation has exactly the same solutions as the original system of equations, and this completes the proof. \qed
\end{proof}



\section{Conclusions}
\label{sec:conclusions}

The quest for a suitable algebraic abstraction of dynamical systems where polynomial equations are tractable, such as a semiring~$R$ with a surjective homomorphism~$\DynamicalSystems \to R$ that does not erase too much information, is not over. However, we feel like the semiring~$\Profiles$ itself still deserves further investigation. Is the borderline between decidable and undecidable equation problems, for instance in terms of polynomial degree or number of variables, the same as for natural numbers? Are there interesting subclasses of equations that are solvable in polynomial time, and others decidable but strictly harder than~$\NP$? Is there a polynomial-time reducibility test? And, from a more algebraic perspective, what are the prime elements of~$\Profiles$? Do they exist at all?

\subsubsection{Acknowledgements}

Caroline Gaze-Maillot was funded by a research internship and Antonio E. Porreca by his salary of French public servant (both affiliated to Aix Marseille Université, Université de Toulon, CNRS, LIS, Marseille, France). This work is an extended version of Caroline Gaze-Maillot's research internship work~\cite{GazeMaillot2020a}. We would like to thank Luca Manzoni for several fruitful discussions about the subject of this paper, in particular on the generating sets of~$\Profiles$ as an~$\N$-semimodule (Theorem~\ref{thm:profiles-generators}), Florian Bridoux for having the good idea on how to reduce several linear equations to a single one (Theorem~\ref{thm:one-equation-np-complete}), and Ananda Ayu Permatasari for finding an error in the original proof of Theorem~\ref{thm:majority-irreducible}.

\bibliographystyle{splncs04}
\bibliography{Bibliography}

\end{document}